\documentclass[11pt,a4paper]{article}

\usepackage{graphicx}
\usepackage{amsthm,amsmath,amssymb}
\usepackage{algorithm,algorithmic}
\usepackage{array}
\usepackage{color}
\usepackage{enumerate}
\usepackage{multirow}
\usepackage[normalem]{ulem}
\usepackage{subfig}
\usepackage{authblk}
\usepackage{sectsty}
\usepackage{bm}
\usepackage{booktabs}
\usepackage{cite}

\newcommand{\nop}[1]{}

\newcommand{\OPT}{\textsf{OPT}}
\newcommand{\OPTc}{\textsf{OPT}_\textsf{cut}}
\newcommand{\SDPc}{\textsf{SDP}_\textsf{cut}}

\DeclareMathOperator*{\argmin}{arg\,min}

\theoremstyle{definition}
\newtheorem{theorem}     {Theorem}
\newtheorem{lemma}       {Lemma}

\newtheorem{corollary}   {Corollary}

\newtheorem{remark}      {Remark}
\newtheorem{fact}        {Fact}

\title{Additive Approximation Algorithms for Modularity Maximization}

\author{Yasushi Kawase\thanks{email: kawase.y.ab@m.titech.ac.jp}}
\author{Tomomi Matsui\thanks{email: matsui.t.af@m.titech.ac.jp}}
\author{Atsushi Miyauchi\thanks{email: miyauchi.a.aa@m.titech.ac.jp (Corresponding author)}}
\affil{\small{\textit{Graduate School of Decision Science and Technology, Tokyo Institute of Technology, \\
\textit{Ookayama 2-12-1, Meguro-ku, Tokyo 152-8552, Japan}}}}

\date{\today}

\begin{document}
\maketitle

\begin{abstract}
The modularity is a quality function in community detection, which was introduced by Newman and Girvan~(2004). 
Community detection in graphs is now often conducted through modularity maximization: 
given an undirected graph $G=(V,E)$, we are asked to find a partition $\mathcal{C}$ of $V$ that maximizes the modularity. 
Although numerous algorithms have been developed to date, most of them have no theoretical approximation guarantee. 
Recently, to overcome this issue, the design of modularity maximization algorithms with provable approximation guarantees 
has attracted significant attention in the computer science community.

In this study, we further investigate the approximability of modularity maximization. 
More specifically, we propose a polynomial-time 
$\left(\cos\left(\frac{3-\sqrt{5}}{4}\pi\right)-\frac{1+\sqrt{5}}{8}\right)$-additive approximation algorithm 
for the modularity maximization problem. 
Note here that $\cos\left(\frac{3-\sqrt{5}}{4}\pi\right)-\frac{1+\sqrt{5}}{8} < 0.42084$ holds. 
This improves the current best additive approximation error of 0.4672, which was recently provided by Dinh, Li, and Thai (2015). 
Interestingly, our analysis also demonstrates that the proposed algorithm obtains a nearly-optimal solution for any instance with a very high modularity value. 
Moreover, we propose a polynomial-time 0.16598-additive approximation algorithm for the maximum modularity cut problem. 
It should be noted that this is the first non-trivial approximability result for the problem. 
Finally, we demonstrate that our approximation algorithm can be extended to some related problems. 
\end{abstract}

\section{Introduction}
Identifying community structure is a fundamental primitive in graph mining~\cite{Fo10}. 
Roughly speaking, a \textit{community} (also referred to as a \textit{cluster} or \textit{module}) in a graph 
is a subset of vertices densely connected with each other, but sparsely connected with the vertices outside the subset. 
Community detection in graphs is a powerful way to discover components that have some special roles or possess important functions. 
For example, consider the graph representing the World Wide Web, 
where vertices correspond to web pages and edges represent hyperlinks between pages. 
Communities in this graph are likely to be the sets of web pages dealing with the same or similar topics, 
or sometimes link spam~\cite{Gi05}. 
As another example, consider the protein--protein interaction graphs, 
where vertices correspond to proteins within a cell and edges represent interactions between proteins. 
Communities in this graph are likely to be the sets of proteins that have the same or similar functions within the cell~\cite{SpMi03}.

To date, numerous community detection algorithms have been developed, 
most of which are designed to maximize a \textit{quality function}. 
Quality functions in community detection return some value 
that represents the \textit{community-degree} for a given partition of the set of vertices. 
The best known and widely used quality function is the \textit{modularity}, 
which was introduced by Newman and Girvan~\cite{NeGi04}.
Let $G=(V,E)$ be an undirected graph consisting of $n=|V|$ vertices and $m=|E|$ edges. 
The modularity, a quality function for a partition $\mathcal{C}=\{C_1,\dots,C_k\}$ of $V$ 
(i.e., $\bigcup_{i=1}^k C_i=V$ and $C_i\cap C_j=\emptyset$ for $i\neq j$), can be written as 
\begin{align*}
Q(\mathcal{C})=\sum_{C\in \mathcal{C}}\left(\frac{m_C}{m} - \left(\frac{D_C}{2m}\right)^2\right), 
\end{align*}
where $m_C$ represents the number of edges whose endpoints are both in $C$, 
and $D_C$ represents the sum of degrees of the vertices in $C$. 
The modularity represents the sum, over all communities, of the fraction of the number of edges within communities 
minus the expected fraction of such edges assuming that they are placed at random with the same degree distribution. 
The modularity is based on the idea that the greater the above surplus, the more community-like the partition $\mathcal{C}$. 

Although the modularity is known to have some drawbacks 
(e.g., the \textit{resolution limit}~\cite{FoBa07} and \textit{degeneracies}~\cite{GoMoCl10}), 
community detection is now often conducted through modularity maximization: 
given an undirected graph $G=(V, E)$, we are asked to find a partition $\mathcal{C}$ of $V$ that maximizes the modularity. 
Note that the modularity maximization problem has no restriction on the number of communities in the output partition; 
thus, the algorithms are allowed to specify the best number of communities by themselves. 
Brandes et al.~\cite{Bretal08} proved that modularity maximization is NP-hard. 
This implies that unless $\text{P}=\text{NP}$, 
there exists no polynomial-time algorithm that finds a partition with maximum modularity for any instance. 
A wide variety of applications (and this hardness result) have promoted the development of modularity maximization heuristics. 
In fact, there are numerous algorithms based on various techniques 
such as greedy procedure~\cite{Bletal08,ClNeMo04,NeGi04}, simulated annealing~\cite{GuAm05,MaDo05}, 
spectral optimization~\cite{Ne06-1,RiMuPo09}, extremal optimization~\cite{DuAr05}, 
and mathematical programming~\cite{AgKe08,CaHaLi11,MiMi13,CaCoHa14}.
Although some of them are known to perform well in practice, they have no theoretical approximation guarantee at all. 

Recently, to overcome this issue, the design of modularity maximization algorithms with provable approximation guarantees 
has attracted significant attention in the computer science community. 
DasGupta and Desai~\cite{DaDe13} designed a polynomial-time $\epsilon$-additive approximation algorithm\footnote{
A feasible solution is  \textit{$\alpha$-additive approximate} if its objective value is at least the optimal value minus $\alpha$. 
An algorithm is called an \textit{$\alpha$-additive approximation algorithm} 
if it returns an $\alpha$-additive approximate solution for any instance.
For an $\alpha$-additive approximation algorithm, $\alpha$ is referred to as an \textit{additive approximation error} of the algorithm.} 
for dense graphs (i.e., graphs with $m=\Omega(n^2)$) using an algorithmic version of the regularity lemma~\cite{FrKa96}, 
where $\epsilon>0$ is an arbitrary constant. 
Moreover, Dinh, Li, and Thai~\cite{DiLiTh15} very recently developed a polynomial-time 0.4672-additive approximation algorithm. 
This is the first polynomial-time additive approximation algorithm with a non-trivial approximation guarantee
for modularity maximization (that is applicable to any instance).\footnote{
A 1-additive approximation algorithm is trivial because $Q(\{V\})=0$ and $Q(\mathcal{C})< 1$ for any partition $\mathcal{C}$.} 
Note that, to our knowledge, this is the current best additive approximation error. 
Their algorithm is based on the semidefinite programming (SDP) relaxation and the hyperplane separation technique.

\subsection{Our Contribution}
In this study, we further investigate the approximability of modularity maximization. 
Our contribution can be summarized as follows: 

\begin{enumerate}
\item We propose a polynomial-time $\left(\cos\left(\frac{3-\sqrt{5}}{4}\pi\right)-\frac{1+\sqrt{5}}{8}\right)$-additive approximation algorithm 
for the modularity maximization problem. 
Note here that $\cos\left(\frac{3-\sqrt{5}}{4}\pi\right)-\frac{1+\sqrt{5}}{8} < 0.42084$ holds; 
thus, this improves the current best additive approximation error of 0.4672, 
which was recently provided by Dinh, Li, and Thai~\cite{DiLiTh15}. 
Interestingly, our analysis also demonstrates that the proposed algorithm obtains a nearly-optimal solution for any instance with a very high modularity value. 
\item We propose a polynomial-time 0.16598-additive approximation algorithm for the maximum modularity cut problem. 
It should be noted that this is the first non-trivial approximability result for the problem. 
\item We demonstrate that our additive approximation algorithm for the modularity maximization problem can be extended to some related problems. 
\end{enumerate}

\paragraph{First result.}
Let us describe our first result in details. 
Our additive approximation algorithm is also based on the SDP relaxation and the hyperplane separation technique. 
However, as described below, our algorithm is essentially different from the one proposed by Dinh, Li, and Thai~\cite{DiLiTh15}. 

The algorithm by Dinh, Li, and Thai~\cite{DiLiTh15} reduces 
the SDP relaxation for the modularity maximization problem to the one for \textsc{MaxAgree} problem arising in correlation clustering 
(e.g., see \cite{BaBlCh04} or \cite{ChGuWi05}) by adding an appropriate constant to the objective function. 
Then, the algorithm adopts the SDP-based $0.7664$-approximation algorithm\footnote{
A feasible solution is \textit{$\alpha$-approximate} if its objective value is at least $\alpha$ times the optimal value. 
An algorithm is called an \textit{$\alpha$-approximation algorithm} if it returns an $\alpha$-approximate solution for any instance.
For an $\alpha$-approximation algorithm, $\alpha$ is referred to as an \textit{approximation ratio} of the algorithm.}  
for \textsc{MaxAgree} problem, 
which was developed by Charikar, Guruswami, and Wirth~\cite{ChGuWi05}. 
In fact, the additive approximation error of $0.4672$ is just derived from $2(1-\kappa)$, 
where $\kappa$ represents the approximation ratio of the SDP-based algorithm for \textsc{MaxAgree} problem (i.e., $\kappa=0.7664$). 
It should be noted that the analysis of the SDP-based algorithm for \textsc{MaxAgree} problem~\cite{ChGuWi05} 
aims at multiplicative approximation rather than additive one. 
As a result, the analysis by Dinh, Li, and Thai~\cite{DiLiTh15} has caused a gap in terms of additive approximation. 

In contrast, our algorithm does not depend on such a reduction. 
In fact, our algorithm just solves the SDP relaxation for the modularity maximization problem without any transformation. 
Moreover, our algorithm employs a hyperplane separation procedure different from the one used in their algorithm. 
The algorithm by Dinh, Li, and Thai~\cite{DiLiTh15} generates 2 and 3 random hyperplanes to obtain feasible solutions, and then returns the better one. 
On the other hand, our algorithm chooses an appropriate number of hyperplanes using the information of the optimal solution to the SDP relaxation 
so that the lower bound on the expected modularity value is maximized. 
Note here that this modification does not improve the worst-case performance of the algorithm by Dinh, Li, and Thai~\cite{DiLiTh15}. 
In fact, as shown in our analysis, their algorithm already has the additive approximation error of 
$\cos\left(\frac{3-\sqrt{5}}{4}\pi\right)-\frac{1+\sqrt{5}}{8}$. 
However, we demonstrate that the proposed algorithm has a much better lower bound on the expected modularity value for many instances. 
In particular, for any instance with optimal value close to 1 (a trivial upper bound), our algorithm obtains a nearly-optimal solution. 
At the end of our analysis, we summarize a lower bound on the expected modularity value with respect to the optimal value of a given instance. 


\paragraph{Second result.}
Here we describe our second result in details. 
The modularity maximization problem has no restriction on the number of clusters in the output partition. 
On the other hand, there also exist a number of problem variants with such a restriction. 
The maximum modularity cut problem is a typical one, where given an undirected graph $G=(V,E)$, 
we are asked to find a partition $\mathcal{C}$ of $V$ consisting of at most two components 
(i.e., a bipartition $\mathcal{C}$ of $V$) that maximizes the modularity. 
This problem appears in many contexts in community detection. 
For example, a few hierarchical divisive heuristics for the modularity maximization problem repeatedly solve this problem 
either exactly~\cite{CaCoHa14,CaHaLi11} or heuristically~\cite{AgKe08}, to obtain a partition $\mathcal{C}$ of $V$. 
Brandes et al.~\cite{Bretal08} proved that the maximum modularity cut problem is NP-hard (even on dense graphs). 
More recently, DasGupta and Desai~\cite{DaDe13} showed that the problem is NP-hard even on $d$-regular graphs with any fixed $d\geq 9$. 
However, to our knowledge, there exists no approximability result for the problem. 

Our additive approximation algorithm adopts the SDP relaxation and the hyperplane separation technique, 
which is identical to the subroutine of the hierarchical divisive heuristic proposed by Agarwal and Kempe~\cite{AgKe08}. 
Specifically, our algorithm first solves the SDP relaxation for the maximum modularity cut problem (rather than the modularity maximization problem), 
and then generates a random hyperplane to obtain a feasible solution for the problem. 
Although the computational experiments by Agarwal and Kempe~\cite{AgKe08} demonstrate that 
their hierarchical divisive heuristic maximizes the modularity quite well in practice, 
the approximation guarantee of the subroutine in terms of the maximum modularity cut was not analyzed. 
Our analysis shows that the proposed algorithm is a $0.16598$-additive approximation algorithm for the maximum modularity cut problem. 
At the end of our analysis, we again present a lower bound on the expected modularity value with respect to the optimal value of a given instance. 
This reveals that for any instance with optimal value close to $1/2$ (a trivial upper bound in the case of bipartition), 
our algorithm obtains a nearly-optimal solution.

\paragraph{Third result.}
Finally, we describe our third result. 
In addition to the above problem variants with a bounded number of clusters, 
there are many other variations of modularity maximization~\cite{Fo10}. 
We demonstrate that our additive approximation algorithm for the modularity maximization problem can be extended to the following three problems: 
the weighted modularity maximization problem~\cite{Ne04}, the directed modularity maximization problem~\cite{LeNe08}, 
and Barber's bipartite modularity maximization problem~\cite{Ba07}.


\subsection{Related Work}
\paragraph{SDP relaxation.}
The seminal work by Goemans and Williamson~\cite{GoWi95} has opened the door to the design of approximation algorithms 
using the SDP relaxation and the hyperplane separation technique. 
To date, this approach has succeeded in developing approximation algorithms for various NP-hard problems~\cite{WiSh11}. 
As mentioned above, Agarwal and Kempe~\cite{AgKe08} introduced the SDP relaxation for the maximum modularity cut problem. 
For the original modularity maximization problem, the SDP relaxation was recently used by Dinh, Li, and Thai~\cite{DiLiTh15}.

\paragraph{Multiplicative approximation algorithms.}
As mentioned above, the design of approximation algorithms for modularity maximization has recently become 
an active research area in the computer science community. 
Indeed, in addition to the additive approximation algorithms described above, 
there also exist multiplicative approximation algorithms. 

DasGupta and Desai~\cite{DaDe13} designed an $\Omega(1/\log d)$-approximation algorithm for $d$-regular graphs with $d\leq \frac{n}{2\log n}$. 
Moreover, they developed an approximation algorithm for the weighted modularity maximization problem. 
The approximation ratio is logarithmic in the maximum weighted degree of edge-weighted graphs 
(where the edge-weights are normalized so that the sum of weights are equal to the number of edges). 
This algorithm requires that the maximum weighted degree is less than about $\frac{\sqrt[5]{n}}{\log n}$. 
These algorithms are not derived directly from logarithmic approximation algorithms for quadratic forms 
(e.g., see \cite{Aletal05} or \cite{ChWi04}) 
because the quadratic form for modularity maximization has negative diagonal entries. 
To overcome this difficulty, they designed a more specialized algorithm using a graph decomposition technique. 

Dinh and Thai~\cite{DiTh13} developed multiplicative approximation algorithms for the modularity maximization problem on scale-free graphs with a prescribed degree sequence. 
In their graphs, the number of vertices with degree $d$ is fixed to some value proportional to $d^{-\gamma}$, where $-\gamma$ is called the power-law exponent. 
For such scale-free graphs with $\gamma>2$, they developed a polynomial-time 
$\left(\frac{\zeta(\gamma)}{\zeta(\gamma-1)}-\epsilon \right)$-approximation algorithm for an arbitrarily small $\epsilon>0$, 
where $\zeta(\gamma)=\sum_{i=1}^\infty \frac{1}{i^\gamma}$ is the Riemann zeta function. 
For graphs with $1<\gamma \leq 2$, they developed a polynomial-time $\Omega(1/\log n)$-approximation algorithm 
using the logarithmic approximation algorithm for quadratic forms~\cite{ChWi04}.

\paragraph{Inapproximability results.} 
There are some inapproximability results for the modularity maximization problem. 
DasGupta and Desai~\cite{DaDe13} showed that it is NP-hard to obtain a $(1-\epsilon)$-approximate solution 
for some constant $\epsilon>0$ (even for complements of 3-regular graphs). 
More recently, Dinh, Li, and Thai~\cite{DiLiTh15} proved a much stronger statement, that is, 
there exists no polynomial-time $(1-\epsilon)$-approximation algorithm for \textit{any} $\epsilon>0$, unless $\text{P}=\text{NP}$. 
It should be noted that these results are on multiplicative approximation rather than additive one.  
In fact, there exist no inapproximability results in terms of additive approximation for modularity maximization.

\subsection{Preliminaries}
Here we introduce definitions and notation used in this paper. 
Let $G=(V,E)$ be an undirected graph consisting of $n=|V|$ vertices and $m=|E|$ edges. 
Let $P=V\times V$. 
By simple calculation, as mentioned in Brandes et al.~\cite{Bretal08}, the modularity can be rewritten as 
\begin{align*}
Q(\mathcal{C}) = \frac{1}{2m}\sum_{(i,j)\in P}\left(A_{ij} - \frac{d_id_j}{2m}\right)\delta(\mathcal{C}(i), \mathcal{C}(j)),
\end{align*}
where $A_{ij}$ is the $(i,j)$ component of the adjacency matrix $A$ of $G$, $d_i$ is the degree of $i\in V$, 
$\mathcal{C}(i)$ is the (unique) community to which $i\in V$ belongs, 
and $\delta$ represents the Kronecker symbol equal to $1$ if two arguments are identical and $0$ otherwise. 
This form is useful to write mathematical programming formulations for modularity maximization. 
For convenience, we define  
\begin{align*}
q_{ij}=\frac{A_{ij}}{2m} - \frac{d_id_j}{4m^2}\quad \text{for each } (i,j)\in P. 
\end{align*}
We can divide the set $P$ into the following two disjoint subsets: 
\begin{align*}
P_{\geq 0}= \{(i,j)\in P\mid q_{ij}\geq 0\} \quad \text{and} \quad P_{<0}= \{(i,j)\in P\mid q_{ij}<0\}.
\end{align*}
Clearly, we have 
\begin{align*}
\sum_{(i,j)\in P_{\geq 0}}q_{ij} + \sum_{(i,j)\in P_{<0}}q_{ij} =\sum_{(i,j)\in P}q_{ij}= 0,
\end{align*} 
and thus 
\begin{align*}
\sum_{(i,j)\in P_{\geq 0}}q_{ij}=\sum_{(i,j)\in P_{<0}}-q_{ij}.
\end{align*} 
We denote this value by $q$, i.e., 
\begin{align*}
q=\sum_{(i,j)\in P_{\geq 0}}q_{ij}. 
\end{align*}
Note that for any instance, we have $q<1$.

\subsection{Paper Organization}
This paper is structured as follows. 
In Section~\ref{sec:algorithm}, we revisit the SDP relaxation for the modularity maximization problem, and then describe an outline of our algorithm.
In Section~\ref{sec:analysis}, the approximation guarantee of the proposed algorithm is carefully analyzed. 
In Section~\ref{sec:cut}, we propose an additive approximation algorithm for the maximum modularity cut problem. 
We extend our additive approximation algorithm to some related problems in Section~\ref{sec:related}. 
Finally, conclusions and future work are presented in Section~\ref{sec:conclusions}.

\section{Algorithm}\label{sec:algorithm}
In this section, we revisit the SDP relaxation for the modularity maximization problem, and then describe an outline of our algorithm. 
The modularity maximization problem can be formulated as follows:
\begin{alignat*}{3}
&\text{Maximize}   &\ \  &\sum_{(i,j)\in P}q_{ij}\, (\bm{y}_i\cdot \bm{y}_j)\\
&\text{subject to} &     &\bm{y}_i\in \{\bm{e}_1,\dots,\bm{e}_n\} \quad \quad (\forall i\in V),
\end{alignat*}
where $\bm{e}_k$ $(1\leq k\leq n)$ represents the vector that has 1 in the $k$th coordinate and 0 elsewhere. 
We denote by $\OPT$ the optimal value of this original problem. 
Note that for any instance, we have $\OPT\in [0,1)$.
We introduce the following semidefinite relaxation problem: 
\begin{alignat*}{4}
&\textsf{SDP}:&\ \  &\text{Maximize}    &\ \ &\sum_{(i,j)\in P}q_{ij}x_{ij}\\
&             &      &\text{subject to} &    &x_{ii}= 1                       &  &(\forall i\in V) ,\\
&             &      &                  &    &x_{ij}\geq 0                    &  &(\forall i,j\in V),\\
&             &      &                  &    &X=(x_{ij}) \in \mathcal{S}^n_+,
\end{alignat*}
where $\mathcal{S}^n_+$ represents the cone of $n\times n$ symmetric positive semidefinite matrices. 
It is easy to see that every feasible solution $X=(x_{ij})$ of \textsf{SDP} satisfies $x_{ij}\leq 1$ for any $(i,j)\in P$. 
Although the algorithm by Dinh, Li, and Thai~\cite{DiLiTh15} reduces \textsf{SDP} 
to the one for \textsc{MaxAgree} problem by adding an appropriate constant to the objective function, 
our algorithm just solves \textsf{SDP} without any transformation. 
Let $X^*=(x^*_{ij})$ be an optimal solution to \textsf{SDP}, 
which can be computed (with an arbitrarily small error) in time polynomial in $n$ and $m$. 
Using the optimal solution $X^*$, we define the following two values: 
\begin{align*}
z^*_+= \frac{1}{q} \sum_{(i,j)\in P_{\geq 0}} q_{ij} x^*_{ij}
\quad \mbox{and} \quad 
z^*_-= \frac{1}{q}\sum_{(i,j)\in P_{<0}} q_{ij} x^*_{ij}, 
\end{align*}
both of which are useful in the analysis of the approximation guarantee of our algorithm. 
Clearly, we have $0\leq z^*_+\leq 1$ and $-1\leq z^*_-\leq 0$. 

We apply the hyperplane separation technique to obtain a feasible solution of the modularity maximization problem.
Specifically, we consider the following general procedure: generate $k$ random hyperplanes to separate the vectors 
corresponding to the optimal solution $X^*$, and then obtain a partition $\mathcal{C}_k=\{C_1,\dots, C_{2^k}\}$ of $V$. 
For reference, the procedure is described in Algorithm~\ref{alg:hyperplane}. 
Note here that at this time, we have not yet mentioned how to determine the number $k$ of hyperplanes we generate. 
As revealed in our analysis, we can choose an appropriate number of hyperplanes using the value of $z^*_+$ 
so that the lower bound on the expected modularity value of the output of \textsf{Hyperplane}($k$) is maximized. 
\begin{algorithm}[t]
\caption{\textsf{Hyperplane}($k$)}
\label{alg:hyperplane}
\begin{algorithmic}[1]
\REQUIRE Graph $G=(V,E)$
\ENSURE Partition $\mathcal{C}$ of $V$
\STATE Obtain an optimal solution $X^*=(x^*_{ij})$ to \textsf{SDP}
\STATE Generate $k$ random hyperplanes and obtain a partition $\mathcal{C}_k=\{C_1,\dots, C_{2^k}\}$ of $V$
\STATE \textbf{return} $\mathcal{C}_k$
\end{algorithmic}
\end{algorithm}

\section{Analysis}\label{sec:analysis}
In this section, we first analyze an additive approximation error of \textsf{Hyperplane}($k$) for each positive integer $k\in \mathbb{Z}_{>0}$.  
Then, we provide an appropriate number $k^*\in \mathbb{Z}_{>0}$ of hyperplanes we generate, which completes the design of our algorithm. 
Finally, we present a lower bound on the expected modularity value of the output of \textsf{Hyperplane}($k^*$) with respect to the value of \OPT. 

When $k$ random hyperplanes are generated independently, 
the probability that two vertices $i,j\in V$ are in the same cluster is given by 
\begin{align*}
\left(1-\frac{\arccos(x^*_{ij})}{\pi}\right)^k,
\end{align*}
as mentioned in previous works (e.g., see \cite{ChGuWi05} or \cite{GoWi95}).
For simplicity, we define the function 
\begin{align*}
f_k (x) = \left(1-\frac{\arccos(x)}{\pi}\right)^k
\end{align*}
for $x\in [0,1]$. 
Here we present the lower convex envelope of each of  $f_k(x)$ and $-f_k(x)$. 
\begin{lemma}\label{lem:lce}
For any positive integer $k$, the lower convex envelope of $f_k(x)$ is given by $f_k(x)$ itself, 
and the lower convex envelope of $-f_k(x)$ is given by the linear function $h_k(x)=-1/2^k+(1/2^k-1)x $ for $x\in [0,1]$. 
\end{lemma}

The following lemma lower bounds the expected modularity value of the output of \textsf{Hyperplane}($k$).
\begin{lemma}\label{lem:lb}
Let $\mathcal{C}_k$ be the output of \textsf{Hyperplane}($k$). For any positive integer $k$, it holds that 
\begin{align*}
\text{E}[Q(\mathcal{C}_k)]\geq q\left(f_k (z^*_+) +h_k (-z^*_-)\right).
\end{align*}
\end{lemma}
\begin{proof}
Recall that $\mathcal{C}_k(i)$ for each $i\in V$ denotes the (unique) cluster in $\mathcal{C}_k$ that includes the vertex $i$.
Note here that $\delta(\mathcal{C}_k(i),\mathcal{C}_k(j))$ for each $(i,j)\in P$ is a random variable, 
which takes 1 with probability $f_k(x^*_{ij})$ and 0 with probability $1-f_k(x^*_{ij})$. 
The expectation $\text{E}[Q(\mathcal{C}_k)]$ is lower bounded as follows: 
\begin{align*}
\text{E}[Q(\mathcal{C}_k)] &= \text{E}\left[\sum_{(i,j)\in P}q_{ij}\delta(\mathcal{C}_k(i),\mathcal{C}_k(j))\right]\\
&=\sum_{(i,j)\in P}q_{ij}f_k(x^*_{ij})\\
&=\sum_{(i,j)\in P_{\geq 0}}q_{ij}f_k(x^*_{ij}) + \sum_{(i,j)\in P_{<0}}-q_{ij}\cdot (-f_k(x^*_{ij}))\\
&\geq \sum_{(i,j)\in P_{\geq 0}}q_{ij}f_k(x^*_{ij}) + \sum_{(i,j)\in P_{<0}}-q_{ij}h_k(x^*_{ij})\\
&=q\left(\sum_{(i,j)\in P_{\geq 0}}\left(\frac{q_{ij}}{q}\right)f_k(x^*_{ij}) + \sum_{(i,j)\in P_{<0}}\left(\frac{-q_{ij}}{q}\right)h_k(x^*_{ij})\right)\\
&\geq q\left(f_k\left(\sum_{(i,j)\in P_{\geq 0}}\left(\frac{q_{ij}}{q}\right)x^*_{ij}\right) + h_k\left(\sum_{(i,j)\in P_{<0}}\left(\frac{-q_{ij}}{q}\right)x^*_{ij}\right)\right)\\
&=q\left(f_k(z^*_+)+h_k(-z^*_-)\right),
\end{align*}
where the last inequality follows from Jensen's inequality. 
\end{proof}

The following lemma provides an additive approximation error of \textsf{Hyperplane}($k$) 
by evaluating the above lower bound on $\text{E}[Q(\mathcal{C}_k)]$ using the value of $\OPT$.
\begin{lemma}\label{lem:lb_opt}
For any positive integer $k$, it holds that
\begin{align*}
\text{E}[Q(\mathcal{C}_k)]\geq \OPT-q\left(z^*_+ -f_k(z^*_+)+\frac{1}{2^k}\right).
\end{align*}
\end{lemma}

\begin{proof}
Clearly, $(z^*_+,z^*_-)$ satisfies $q(z^*_+ +z^*_-)\geq \OPT$ and $z^*_-\leq 0$.
Thus, we obtain
\begin{align*}
q\left(f_k (z^*_+) +h_k (-z^*_-)\right)
&\geq \left(\OPT - q(z^*_+ +z^*_-)\right)+ q\left(f_k(z^*_+)+h_k(-z^*_-)\right) \\
&= \left(\OPT - q(z^*_+ +z^*_-)\right)+ q\left(f_k(z^*_+)-1/2^k +(1/2^k-1)(-z^*_-)\right) \\
&= \OPT - q\left(z^*_+ -f_k(z^*_+)+1/2^k +(1/2^k)z^*_-\right) \\
&\geq \OPT - q\left(z^*_+ -f_k(z^*_+)+1/2^k\right). 
\end{align*}
Combining this with Lemma~\ref{lem:lb}, we have
\begin{align*}
\text{E}[Q(\mathcal{C}_k)]\geq q\left(f_k(z^*_+)+h_k(-z^*_-)\right)\geq \OPT -q\left(z^*_+ -f_k(z^*_+)+\frac{1}{2^k}\right),
\end{align*}
as desired.
\end{proof}

For simplicity, we define the function 
\begin{align*}
g_k(x)=x-f_k(x)+\frac{1}{2^k}
\end{align*}
for $x\in [0,1]$. Then, the inequality of the above lemma can be rewritten as 
\begin{align*}
\text{E}[Q(\mathcal{C}_k)]\geq \OPT -q\cdot g_k(z^*_+).
\end{align*}
Figure~\ref{fig:error} plots the above additive approximation error of \textsf{Hyperplane}($k$) with respect to the value of $z^*_+$. 
\begin{figure}[tb]
\centering
\hspace{25mm}
\includegraphics[width=12.5cm]{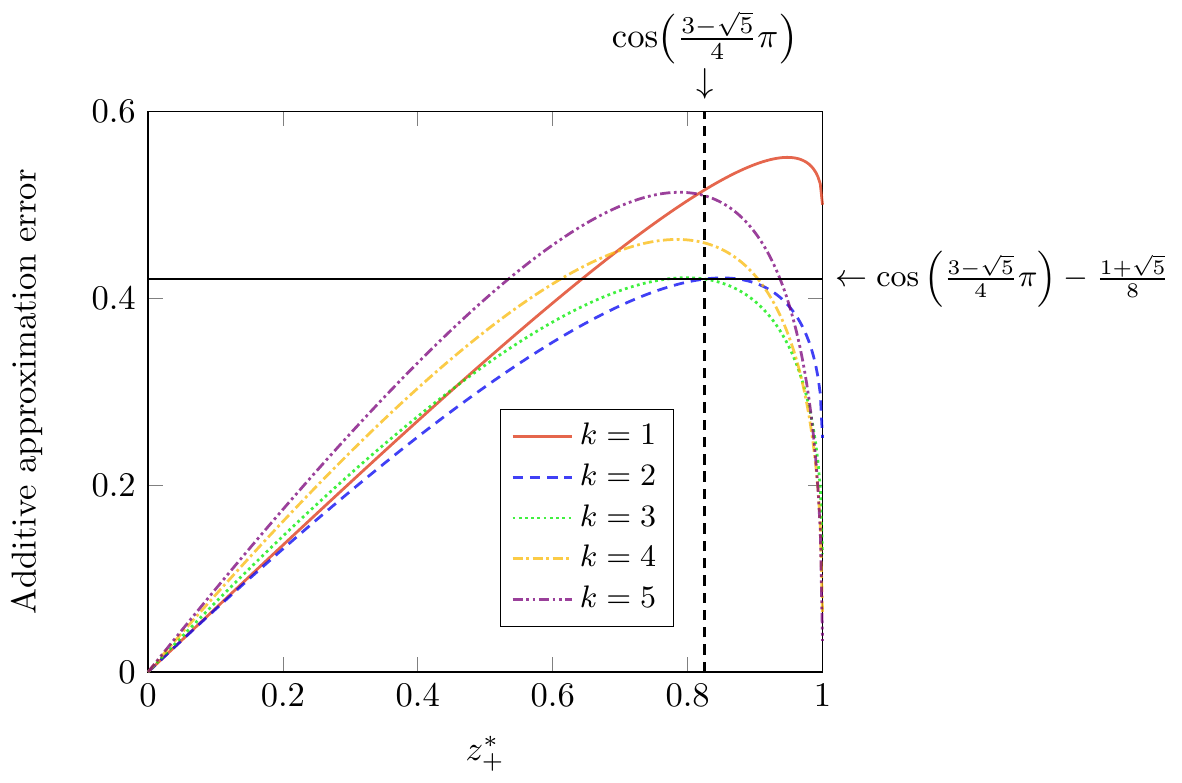}
\caption{A brief illustration of the additive approximation error of \textsf{Hyperplane}($k$) with respect to the value of $z^*_+$. 
For simplicity, we replace $q$ by its upper bound 1. 
Specifically, the function \(g_k(x)=x-f_k(x)+1/2^k\) for $x\in [0,1]$ is plotted for $k=1,2,3,4$, and $5$, as examples. 
The point $\left(\cos\left(\frac{3-\sqrt{5}}{4}\pi\right),\, \cos\left(\frac{3-\sqrt{5}}{4}\pi\right)-\frac{1+\sqrt{5}}{8}\right)$ 
is an intersection of the functions $g_2(x)$ and $g_3(x)$. 
}\label{fig:error}
\end{figure}
As can be seen, the appropriate number of hyperplanes 
(i.e., the number of hyperplanes that minimizes the additive approximation error) depends on the value of $z^*_+$. 
Intuitively, we wish to choose $k^{**}$ that satisfies 
\begin{align*}
k^{**}\in \argmin_{k\in \mathbb{Z}_{>0}}g_k(z^*_+).
\end{align*}
However, it is not clear whether \textsf{Hyperplane}$(k^{**})$ runs in polynomial time. 
In fact, the number $k^{**}$ becomes infinity if the value of $z^*_+$ approaches 1. 
Therefore, alternatively, our algorithm chooses 
\begin{align*}
k^*\in \argmin_{k\in \{1,\dots, \max\{3,\lceil\log_{2}n\rceil\}\}}g_k(z^*_+).
\end{align*}
Our analysis demonstrates that the worst-case performance of \textsf{Hyperplane}$(k^*)$ is exactly the same as that of \textsf{Hyperplane}$(k^{**})$, 
and moreover, the lower bound on the expected modularity value with respect to the value of \textsf{OPT} is not affected by this change. 

The following lemma analyzes the worst-case performance of \textsf{Hyperplane}($k^*$); 
thus, it provides the additive approximation error of \textsf{Hyperplane}($k^*$).
\begin{lemma}\label{lem:minmax}
It holds that
\begin{align*}
\max_{x\in [0,1]}\min_{k\in \{1,\dots, \max\{3,\lceil\log_{2}n\rceil\}\}} g_k(x)
=\cos\left(\frac{3-\sqrt{5}}{4}\pi\right)-\frac{1+\sqrt{5}}{8}.
\end{align*}
\end{lemma}

\begin{proof}
First, we show that 
\begin{align*}
\max_{x\in [0,1]}\min_{k\in \{1,\dots, \max\{3,\lceil\log_{2}n\rceil\}\}} g_k(x)
\geq \cos\left(\frac{3-\sqrt{5}}{4}\pi\right)-\frac{1+\sqrt{5}}{8}.
\end{align*}
It suffices to show that for any $k\in \mathbb{Z}_{>0}$, 
\begin{align*}
g_k\left(\cos\left(\frac{3-\sqrt{5}}{4}\pi\right)\right)\ge \cos\left(\frac{3-\sqrt{5}}{4}\pi\right)-\frac{1+\sqrt{5}}{8}.
\end{align*}
For $k=1,2,3$, and $4$, we have
\begin{align*}
  g_1\left(\cos\left(\frac{3-\sqrt{5}}{4}\pi\right)\right)-\cos\left(\frac{3-\sqrt{5}}{4}\pi\right)+\frac{1+\sqrt{5}}{8}
  &=\frac{1+\sqrt{5}}{8}-\left(\frac{1+\sqrt{5}}{4}\right)+\frac{1}{2}>0,\\
  g_2\left(\cos\left(\frac{3-\sqrt{5}}{4}\pi\right)\right)-\cos\left(\frac{3-\sqrt{5}}{4}\pi\right)+\frac{1+\sqrt{5}}{8}
  &=\frac{1+\sqrt{5}}{8}-\left(\frac{1+\sqrt{5}}{4}\right)^2+\frac{1}{4}
  =0,\\
  g_3\left(\cos\left(\frac{3-\sqrt{5}}{4}\pi\right)\right)-\cos\left(\frac{3-\sqrt{5}}{4}\pi\right)+\frac{1+\sqrt{5}}{8}
  &=\frac{1+\sqrt{5}}{8}-\left(\frac{1+\sqrt{5}}{4}\right)^3+\frac{1}{8}
  =0,
  \end{align*}
and
  \begin{align*}
  g_4\left(\cos\left(\frac{3-\sqrt{5}}{4}\pi\right)\right)-\cos\left(\frac{3-\sqrt{5}}{4}\pi\right)+\frac{1+\sqrt{5}}{8}
  &=\frac{1+\sqrt{5}}{8}-\left(\frac{1+\sqrt{5}}{4}\right)^4+\frac{1}{16}
  >0, 
\end{align*}
respectively. For $k\geq 5$, we have
\begin{align*}
  g_k\left(\cos\left(\frac{3-\sqrt{5}}{4}\pi\right)\right)-\cos\left(\frac{3-\sqrt{5}}{4}\pi\right)+\frac{1+\sqrt{5}}{8}
  &= \frac{1+\sqrt{5}}{8}-\left(\frac{1+\sqrt{5}}{4}\right)^k+\frac{1}{2^k}\\
  &> \frac{1+\sqrt{5}}{8}-\left(\frac{1+\sqrt{5}}{4}\right)^5>0.
\end{align*}
Thus, for any $k\in \mathbb{Z}_{>0}$, we obtain
\begin{align*}
g_k\left(\cos\left(\frac{3-\sqrt{5}}{4}\pi\right)\right)\ge \cos\left(\frac{3-\sqrt{5}}{4}\pi\right)-\frac{1+\sqrt{5}}{8}.
\end{align*}

Next, we show that 
\begin{align*}
\max_{x\in [0,1]}  \min_{k\in \{1,\dots, \max\{3,\lceil\log_{2}n\rceil\}\}} g_k(x)
\leq \cos\left(\frac{3-\sqrt{5}}{4}\pi\right)-\frac{1+\sqrt{5}}{8}.
\end{align*}
By simple calculation, we get
\begin{align*}
g'_2(x)= 1-\frac{2(1-\arccos(x)/\pi)}{\pi\sqrt{1-x^2}}\quad 
\text{and}\quad g'_3(x)=1-\frac{3(1-\arccos(x)/\pi)^2}{\pi\sqrt{1-x^2}}.
\end{align*} 
Let us take an arbitrary $x$ with $0\leq x\leq \cos\left(\frac{3-\sqrt{5}}{4}\pi\right)$.
Since $g'_2(x)>0$, we have  
\begin{align*}
    g_2(x)\leq g_2\left(\cos\left(\frac{3-\sqrt{5}}{4}\pi\right)\right)=\cos\left(\frac{3-\sqrt{5}}{4}\pi\right)-\frac{1+\sqrt{5}}{8}. 
\end{align*}
On the other hand, take $x$ with $\cos\left(\frac{3-\sqrt{5}}{4}\pi\right)\leq x<1$. 
Since $g'_3(x)<0$, we have
\begin{align*}
g_3(x)\leq g_3\left(\cos\left(\frac{3-\sqrt{5}}{4}\pi\right)\right)=\cos\left(\frac{3-\sqrt{5}}{4}\pi\right)-\frac{1+\sqrt{5}}{8}. 
\end{align*}
Note finally that $g_3(1)<\cos\left(\frac{3-\sqrt{5}}{4}\pi\right)-\frac{1+\sqrt{5}}{8}$ holds. Therefore, we obtain
\begin{align*}
\max_{x\in[0,1]}  \min_{k\in \{1,\dots, \max\{3,\lceil\log_{2}n\rceil\}\}} g_k(x)\leq \max_{x\in [0,1]}\min_{k\in \{2,3\}} g_k(x) \leq \cos\left(\frac{3-\sqrt{5}}{4}\pi\right)-\frac{1+\sqrt{5}}{8},
\end{align*}
as desired. 
\end{proof}

\begin{remark}
From the proof of the above lemma, it follows directly that 
\begin{align*}
\max_{x\in [0,1]}\min_{k\in \mathbb{Z}_{>0}} g_k(x)
=\cos\left(\frac{3-\sqrt{5}}{4}\pi\right)-\frac{1+\sqrt{5}}{8}.
\end{align*}
This implies that the worst-case performance of \textsf{Hyperplane}$(k^{**})$ is no better than that of \textsf{Hyperplane}$(k^*)$. 
\end{remark}

\begin{remark}
Here we consider the algorithm that executes \textsf{Hyperplane}($2$) and \textsf{Hyperplane}($3$), 
and then returns the better solution. 
Note that this algorithm is essentially the same as that proposed by Dinh, Li, and Thai~\cite{DiLiTh15}. 
From the proof of the above lemma, it follows immediately that 
\begin{align*}
\max_{x\in [0,1]}\min_{k\in \{2,3\}} g_k(x)
=\cos\left(\frac{3-\sqrt{5}}{4}\pi\right)-\frac{1+\sqrt{5}}{8}.
\end{align*}
This implies that the algorithm by Dinh, Li, and Thai~\cite{DiLiTh15} already has the worst-case performance 
exactly the same as that of \textsf{Hyperplane}($k^*$).  
However, as shown below, \textsf{Hyperplane}($k^*$) has a much better lower bound on the expected modularity value for many instances. 
\end{remark}

Finally, we present a lower bound on the expected modularity value of the output of \textsf{Hyperplane}($k^*$) 
with respect to the value of $\OPT$ (rather than $z^*_+$). 
The following lemma is useful to show that  
the lower bound on the expected modularity value with respect to the value of \textsf{OPT} is not affected by the change from $k^{**}$ to $k^*$. 
The proof can be found in Appendix~\ref{apx:compromise}. 

\begin{lemma}\label{lem:compromise}
For any $k'\in \argmin_{k\in \mathbb{Z}_{>0}} g_k(\OPT)$,
it holds that $k'\leq \max\{3,\lceil\log_{2}n\rceil\}$. 
\end{lemma}

We are now ready to prove the following theorem. 
\begin{theorem}\label{thm:mod}
Let $\mathcal{C}_{k^*}$ be the output of \textsf{Hyperplane}($k^*$). It holds that 
\begin{align*}
\text{E}[Q(\mathcal{C}_{k^*})]\geq \OPT -q\left(\cos\left(\frac{3-\sqrt{5}}{4}\pi\right)-\frac{1+\sqrt{5}}{8}\right). 
\end{align*}
In particular, if $\textsf{OPT}\geq \cos\left(\frac{3-\sqrt{5}}{4}\pi\right)$ holds, then 
\begin{align*}
\text{E}[Q(\mathcal{C}_{k^*})]> \OPT - q\min_{k\in \mathbb{Z}_{>0}} g_k(\OPT). 
\end{align*}
Note here that $q<1$ and $\cos\left(\frac{3-\sqrt{5}}{4}\pi\right)-\frac{1+\sqrt{5}}{8} < 0.42084$. 
\end{theorem}

\begin{proof}
From Lemmas~\ref{lem:lb_opt} and \ref{lem:minmax}, it follows directly that 
\begin{align*}
\text{E}[Q(\mathcal{C}_{k^*})]\geq \OPT -q\left(\cos\left(\frac{3-\sqrt{5}}{4}\pi\right)-\frac{1+\sqrt{5}}{8}\right). 
\end{align*}

Here we prove the remaining part of the theorem. 
Assume that $\textsf{OPT}\geq \cos\left(\frac{3-\sqrt{5}}{4}\pi\right)$ holds. 
By simple calculation, for any $k\in \mathbb{Z}_{>0}$, we have 
\begin{align*}
g''_k(x)= -\frac{k(1-\arccos(x)/\pi)^{k-2}\left((k-1)\sqrt{1-x^2}+\pi x\left(1 - \arccos(x)/\pi\right)\right)}{\pi^2 (1-x^2)^{3/2}},
\end{align*}
all of which are negative for $x\in (0,1)$. 
This means that for any $k\in \mathbb{Z}_{>0}$, the function $g_k(x)$ is strictly concave, 
and moreover, so is the function $\min_{k\in \{1,\dots,\max\{3,\lceil\log_{2}n\rceil\}\}}g_k(x)$.
From the proof of Lemma~\ref{lem:minmax}, 
the function $ \min_{k\in \{1,\dots,\max\{3,\lceil\log_{2}n\rceil\}\}}g_k(x)$ attains its maximum 
(i.e., $\cos\left(\frac{3-\sqrt{5}}{4}\pi\right)-\frac{1+\sqrt{5}}{8}$) 
at $x=\cos\left(\frac{3-\sqrt{5}}{4}\pi\right)$. 
Thus, the function $ \min_{k\in \{1,\dots,\max\{3,\lceil\log_{2}n\rceil\}\}}g_k(x)$ is strictly monotonically decreasing 
over the interval $\left[\cos\left(\frac{3-\sqrt{5}}{4}\pi\right),\, 1\right]$. 
Therefore, we have 
\begin{align*}
	\text{E}[Q(\mathcal{C}_{k^*})] &\geq \OPT -q\min_{k\in \{1,\dots,\max\{3,\lceil\log_{2}n\rceil\}\}} g_k(z^*_+)\\
							   &> \OPT -q  \min_{k\in \{1,\dots,\max\{3,\lceil\log_{2}n\rceil\}\}} g_k(\OPT)\\
							   &= \OPT -q\min_{k\in \mathbb{Z}_{>0}} g_k(\OPT), 
\end{align*}
where the second inequality follows from $z^*_+\geq \OPT/q>\OPT$ and the last equality follows from Lemma~\ref{lem:compromise}. 
\end{proof}

Figure~\ref{fig:lb} depicts the above lower bound on $\text{E}[Q(\mathcal{C}_{k^*})]$. 
As can be seen, if $\OPT$ is close to $1$, then \textsf{Hyperplane}($k^*$) obtains a nearly-optimal solution. 
For example, for any instance with $\OPT \geq 0.99900$, it holds that $\text{E}[Q(\mathcal{C}_{k^*})]> 0.90193$, 
i.e., the additive approximation error is less than 0.09807.  
\begin{figure}[tb]
\centering
\includegraphics[width=9.3cm]{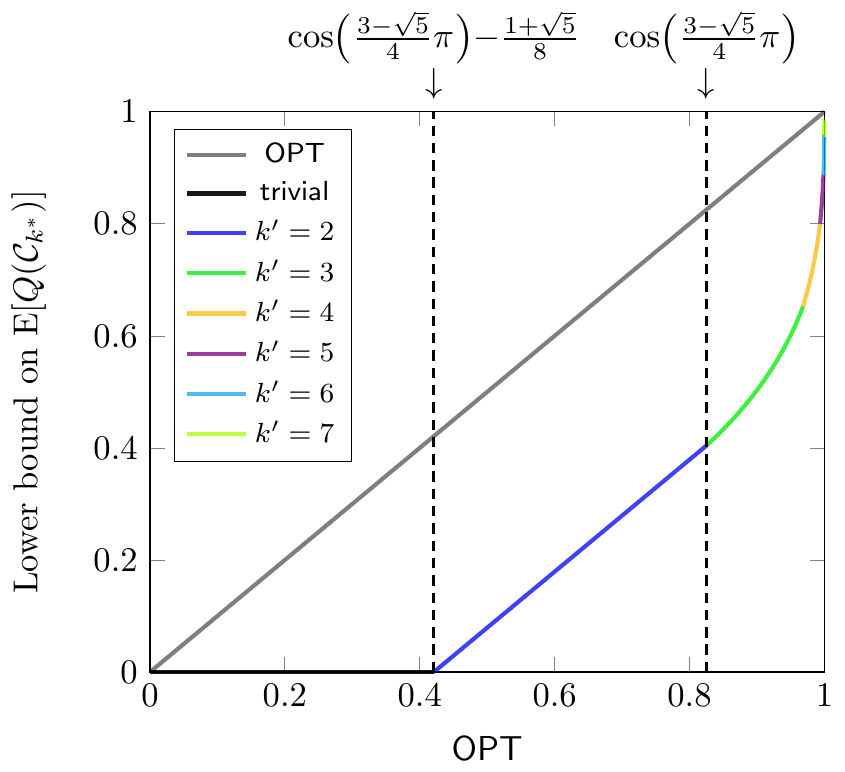}
\caption{A brief illustration of the lower bound on the expected modularity value 
of the output of \textsf{Hyperplane}($k^*$) with respect to the value of \OPT. 
For simplicity, we replace $q$ by its upper bound 1. 
Note that $k'\in \argmin_{k\in \mathbb{Z}_{>0}} g_k(\OPT)$.}\label{fig:lb}
\end{figure}

\begin{remark}
The additive approximation error of \textsf{Hyperplane}($k^*$) depends on the value of $q<1$. 
We see that the less the value of $q$, the better the additive approximation error. 
Thus, it is interesting to find some graphs that have a small value of $q$. 
For instance, for any regular graph $G$ that satisfies $m=\frac{\alpha}{2}n^2$, it holds that $q=1-\alpha$, where $\alpha$ is an arbitrary constant in $(0,1)$.
Here we prove the statement. 
Since $G$ is regular, we have $d_i=2m/n=\alpha n$ for any $i\in V$. 
Moreover, for any $\{i,j\}\in E$, it holds that $q_{ij}=\frac{A_{ij}}{2m}-\frac{d_id_j}{4m^2}=\frac{1}{\alpha n^2}-\frac{1}{n^2}>0$. 
Therefore, we have 
\begin{align*}
q=\sum_{(i,j)\in P_{\geq 0}} q_{ij}=2 \sum_{\{i,j\}\in E}q_{ij}=2m\left(\frac{1}{\alpha n^2}-\frac{1}{n^2}\right)=1-\alpha. 
\end{align*}
\end{remark}

\section{Maximum Modularity Cut}\label{sec:cut}
In this section, we propose a polynomial-time 0.16598-additive approximation algorithm for the maximum modularity cut problem.

\subsection{Algorithm}
In this subsection, we revisit the SDP relaxation for the maximum modularity cut problem, and then describe our algorithm. 
The maximum modularity cut problem can be formulated as follows:
\begin{alignat*}{4}
&\text{Maximize}   &\quad &\frac{1}{4m}\sum_{(i,j)\in P}\left(A_{ij}-\frac{d_id_j}{2m}\right) (y_iy_j+1)\\
&\text{subject to} &      &y_i\in \{-1,1\} \qquad \quad (\forall i\in V).
\end{alignat*}
We denote by $\OPTc$ the optimal value of this original problem. 
Note that for any instance, it holds that $\OPTc \in [0,1/2]$, as shown in DasGupta and Desai~\cite{DaDe13}.
We introduce the following semidefinite relaxation problem: 
\begin{alignat*}{3}
&\SDPc:&\ \ &\text{Maximize}   &\quad &\frac{1}{4m}\sum_{(i,j)\in P}\left(A_{ij}-\frac{d_id_j}{2m}\right) (x_{ij}+1)\\
&&&\text{subject to} &      & x_{ii}=1 \qquad \qquad (\forall i \in V), \\
&&&                  &      & X=(x_{ij})\in \mathcal{S}^n_+, 
\end{alignat*}
where recall that $\mathcal{S}^n_+$ represents the cone of $n\times n$ symmetric positive semidefinite matrices. 
Let $X^*=(x^*_{ij})$ be an optimal solution to $\SDPc$, 
which can be computed (with an arbitrarily small error) in time polynomial in $n$ and $m$. 
Note here that $x^*_{ij}$ may be negative for $(i,j)\in P$ with $i\neq j$, unlike \textsf{SDP} in the previous section. 
The objective function value of $X^*$ can be divided into the following two terms: 
\begin{align*}
z_+^*= \frac{1}{4m} \sum_{(i,j)\in P} A_{ij}(x^*_{ij}+1) 
\quad \text{and} \quad z^*_-= -\frac{1}{8m^2}\sum_{(i,j)\in P} d_id_j(x^*_{ij}+1).
\end{align*}

We generate a random hyperplane to separate the vectors corresponding to the optimal solution $X^*$, 
and then obtain a bipartition $\mathcal{C}=\{C_1,C_2\}$ of $V$. 
For reference, the procedure is described in Algorithm~\ref{alg:cut}. 
As mentioned above, this algorithm is identical to the subroutine of the hierarchical divisive heuristic 
for the modularity maximization problem, which was proposed by Agarwal and Kempe~\cite{AgKe08}. 

\begin{algorithm}[t]
\caption{\textsf{Modularity Cut}}
\label{alg:cut}
\begin{algorithmic}[1]
\REQUIRE Graph $G=(V,E)$
\ENSURE Bipartition $\mathcal{C}$ of $V$ 
\STATE Obtain an optimal solution $X^*=(x^*_{ij})$ to $\SDPc$
\STATE Generate a random hyperplane and obtain a bipartition $\mathcal{C}_\text{out}=\{C_1,C_2\}$ of $V$
\STATE \textbf{return} $\mathcal{C}_\text{out}$
\end{algorithmic}
\end{algorithm}

\subsection{Analysis}
In this subsection, we show that Algorithm~\ref{alg:cut} obtains a 0.16598-additive approximate solution for any instance. 
At the end of our analysis, we present a lower bound on the expected modularity value of the output of Algorithm~\ref{alg:cut} 
with respect to the value of $\OPTc$. 

We start with the following lemma. 
\begin{lemma}\label{lem:z_bound}
It holds that $1/2 \leq z^*_+ \leq 1$ and $-1 \leq z^*_- \leq -1/2$.
\end{lemma}
\begin{proof}
Clearly, $z_+^* \leq 1$ and $z_-^*\geq -1$.
Since $X^*$ is positive semidefinite, we have
\begin{align*}
z_-^*&= - \frac{1}{8m^2} \left(\sum_{(i,j)\in P} d_id_jx^*_{ij} + \sum_{(i,j)\in P} d_id_j \right)
\leq -\frac{1}{8m^2} \left(0 + 4m^2 \right) 
= -\frac{1}{2}. 
\end{align*}
Combining this with $z^*_+ + z^*_-\geq \OPTc \geq 0$, we get 
\begin{align*}
z^*_+\geq \OPTc -z^*_-\geq \OPTc +\frac{1}{2}\geq \frac{1}{2},  
\end{align*}
as desired.
\end{proof}

In Algorithm~\ref{alg:cut}, the probability that two vertices $i,j\in V$ are in the same cluster is given by $1-\arccos(x^*_{ij})/\pi$.
For simplicity, we define the following two functions 
\begin{align*}
p_+(x)=1-\frac{\arccos(x)}{\pi} \quad \text{and}\quad p_-(x)=-\left(1-\frac{\arccos(x)}{\pi}\right)
\end{align*}
for $x\in [-1,1]$ (rather than $x\in [0,1]$). 
Here we present the lower convex envelope of each of $p_+(x)$ and $p_-(x)$. 
\begin{lemma}\label{lem:lce_cut}
Let
\begin{align*}
	\alpha=\min_{-1<x<1} \frac{1-\frac{\arccos (x)}{\pi}}{\frac{x+1}{2}}
			\ (\simeq 0.8785672)  \quad \text{and}\quad 
			\beta = \argmin_{-1<x<1} \frac{1-\frac{\arccos (x)}{\pi}}{\frac{x+1}{2}}
			\ (\simeq 0.6891577).
\end{align*}
The lower convex envelope of $p_+(x)$ is given by  
\begin{align*}
\underline{p_+}(x)=
\begin{cases}
\displaystyle	\alpha \left(\frac{x+1}{2} \right) &(-1 \leq x \leq \beta), \\
\displaystyle	1- \frac{\arccos (x)}{\pi} &(\beta < x \leq 1),
\end{cases} 
\end{align*}
and the lower convex envelope of $p_-(x)$ is given by 
\begin{align*}
\underline{p_-}(x)=
\begin{cases}
\displaystyle	-\left( 1- \frac{\arccos (x)}{\pi} \right) &(-1 \leq x \leq -\beta), \\
\displaystyle	(\alpha-1)- \alpha \left(\frac{x+1}{2}\right) &(-\beta < x \leq 1). 
\end{cases}
\end{align*}
\end{lemma}

The following lemma lower bounds the expected modularity value of the output of Algorithm~\ref{alg:cut}. 

\begin{lemma}\label{lem:lb_cut}
Let $\mathcal{C}_\text{out}$ be the output of Algorithm~\ref{alg:cut}. 
It holds that 
\begin{align*}
\text{E}[Q(\mathcal{C}_\text{out})]\geq \underline{p_+} (2z^*_+-1)+ \underline{p_-}(-1-2z^*_{-}).
\end{align*}
\end{lemma}
\begin{proof}
Recall that $\mathcal{C}_\text{out}(i)$ for each $i\in V$ denotes the (unique) cluster in $\mathcal{C}_\text{out}$ that includes the vertex $i$.
Note here that $\delta(\mathcal{C}_\text{out}(i),\mathcal{C}_\text{out}(j))$ for each $(i,j)\in P$ is a random variable, 
which takes 1 with probability $p_+(x^*_{ij})$ and 0 with probability $1-p_+(x^*_{ij})$. 
The expectation $\text{E}[Q(\mathcal{C}_\text{out})]$ is lower bounded as follows:
\begin{align*}
\text{E}[Q(\mathcal{C}_\text{out})]
&=\text{E}\left[\frac{1}{2m}\sum_{(i,j)\in P}\left(A_{ij}-\frac{d_id_j}{2m}\right)\delta(\mathcal{C}_\text{out}(i),\mathcal{C}_\text{out}(j))\right] \\
&=\text{E}\left[\frac{1}{2m}\sum_{(i,j)\in P} A_{ij}\delta(\mathcal{C}_\text{out}(i),\mathcal{C}_\text{out}(j)) 
-\frac{1}{4m^2}\sum_{(i,j)\in P}d_id_j\delta(\mathcal{C}_\text{out}(i),\mathcal{C}_\text{out}(j))\right] \\
&=\frac{1}{2m}\sum_{(i,j)\in P}A_{ij}p_+(x^*_{ij}) 
+ \frac{1}{4m^2}\sum_{(i,j)\in P}d_id_jp_-(x^*_{ij}) \\
&\geq \frac{1}{2m}\sum_{(i,j)\in P}A_{ij}\underline{p_+}(x^*_{ij}) 
+ \frac{1}{4m^2}\sum_{(i,j)\in P}d_id_j\underline{p_-}(x^*_{ij}) \\
&=\sum_{(i,j)\in P}\left(\frac{A_{ij}}{2m}\right)\underline{p_+}(x^*_{ij}) 
+ \sum_{(i,j)\in P}\left(\frac{d_id_j}{4m^2}\right)\underline{p_-}(x^*_{ij})\\
&\geq \underline{p_+}\left(\sum_{(i,j)\in P}\left(\frac{A_{ij}}{2m}\right)x^*_{ij}\right)
+ \underline{p_-}\left(\sum_{(i,j)\in P}\left(\frac{d_id_j}{4m^2}\right)x^*_{ij}\right)\\
&=\underline{p_+}(2z^*_+-1) + \underline{p_-}(-1-2z^*_-),
\end{align*}
where the last inequality follows from Jensen's inequality. 
\end{proof}

The following lemma provides an additive approximation error of Algorithm~\ref{alg:cut} 
by evaluating the above lower bound on $\text{E}[Q(\mathcal{C}_\text{out})]$ using the value of $\OPTc$. 
\begin{lemma}\label{lem:lb_opt_cut}
It holds that 
\begin{align*}
\text{E}[Q(\mathcal{C}_\text{out})]\geq \OPTc -\left(z^*_+ -\underline{p_+}(2z^*_+ -1) -\frac{\alpha -1}{2}\right).
\end{align*}
\end{lemma}
\begin{proof}
Since $-1 \leq z^*_- \leq -1/2$ by Lemma~\ref{lem:z_bound}, 
we have $0\leq -1-2z^*_{-} \leq 1$ and hence
\begin{align*}
\underline{p_-}(-1-2z^*_{-})
= (\alpha -1) -\alpha \left(\frac{(-1-2z^*_-)+1}{2}\right)
= (\alpha -1) +\alpha z^*_-.
\end{align*}
Thus, recalling $z^*_+ + z^*_-\geq \OPTc$, we obtain
\begin{align*}
\underline{p_+} (2z^*_+-1)+ \underline{p_-}(-1-2z^*_{-})
&\geq (\OPTc - z^*_+ - z^*_-) + \underline{p_+} (2z^*_+-1)+ \underline{p_-}(-1-2z^*_{-})\\
&= (\OPTc - z^*_+ - z^*_-) + \underline{p_+} (2z^*_+-1)+ (\alpha -1)+ \alpha z^*_{-} \\
&= \OPTc - z^*_+  + \underline{p_+} (2z^*_+-1)+ (\alpha -1)+ (\alpha-1) z^*_{-} \\
&\geq \OPTc - z^*_+  + \underline{p_+} (2z^*_+-1)+ (\alpha -1)+ (\alpha-1)(-1/2) \\
&= \OPTc - z^*_+  + \underline{p_+} (2z^*_+-1)+ (\alpha -1)/2. 
\end{align*}
Combining this with Lemma~\ref{lem:lb_cut}, we have
\begin{align*}
\text{E}[Q(\mathcal{C}_\text{out})] 
\geq \underline{p_+} (2z^*_+-1)+ \underline{p_-}(-1-2z^*_{-}) 
\geq \OPTc -\left(z^*_+ -\underline{p_+}(2z^*_+ -1) -\frac{\alpha -1}{2}\right),
\end{align*}
as desired. 
\end{proof}

For simplicity, we define the function 
\begin{align*}
g(x)= x -\underline{p_+} (2x -1) -\frac{\alpha -1}{2}
\end{align*}
for $x\in [1/2,1]$. 
Then, the inequality of the above lemma can be rewritten as 
\begin{align*}
\text{E}[Q(\mathcal{C}_\text{out})]\geq \OPTc -g(z^*_+).
\end{align*}

The following lemma analyzes the worst-case performance of Algorithm~\ref{alg:cut}; 
thus, it provides the additive approximation error of Algorithm~\ref{alg:cut}.
\begin{lemma}\label{lem:max_cut}
It holds that 
\begin{align*}
\max_{1/2\leq x\leq 1} g(x) = \frac{\sqrt{\pi^2-4}}{2\pi} + \frac{\arccos\left(\frac{\sqrt{\pi^2-4}}{\pi}\right)}{\pi}-\frac{\alpha}{2}\ (\simeq 0.165973). 
\end{align*}
\end{lemma}

\begin{proof}
If $1/2\leq x \leq (\beta +1)/2 \ (\simeq 0.844579)$, then $0 \leq 2x-1 \leq \beta$ holds, and hence  
\begin{align*}
g(x)= x -\alpha \left(\frac{(2x-1)+1}{2}\right) -\frac{(\alpha -1)}{2}
= (1-\alpha )\left(x +\frac{1}{2}\right).
\end{align*}
Otherwise (i.e., $(\beta +1)/2 < x\leq 1$), it holds that $\beta < 2 x -1 \leq 1$, and thus 
\begin{align*}
g(x) &=x -\left(1 -\frac{\arccos (2x - 1)}{\pi}\right) -\frac{(\alpha -1)}{2}.
\end{align*}
Summarizing the above, we have 
\begin{align*}
g(x)= 
\begin{cases}
\displaystyle (1-\alpha )\left(x +\frac{1}{2}\right) 
&\displaystyle \left( \frac{1}{2}\leq x \leq \frac{\beta+1}{2}\ (\simeq 0.844579) \right),\\
\displaystyle x -\left(1 -\frac{\arccos(2x - 1)}{\pi}\right) -\frac{(\alpha -1)}{2}
&\displaystyle \left(\frac{\beta+1}{2} < x \leq 1 \right).
\end{cases}
\end{align*}

Figure~\ref{fig:error_cut} plots the above additive approximation error of Algorithm~\ref{alg:cut} with respect to the value of $z^*_+$. 
\begin{figure}[tb]
\centering
\hspace{15mm}
\includegraphics[width=10.8cm]{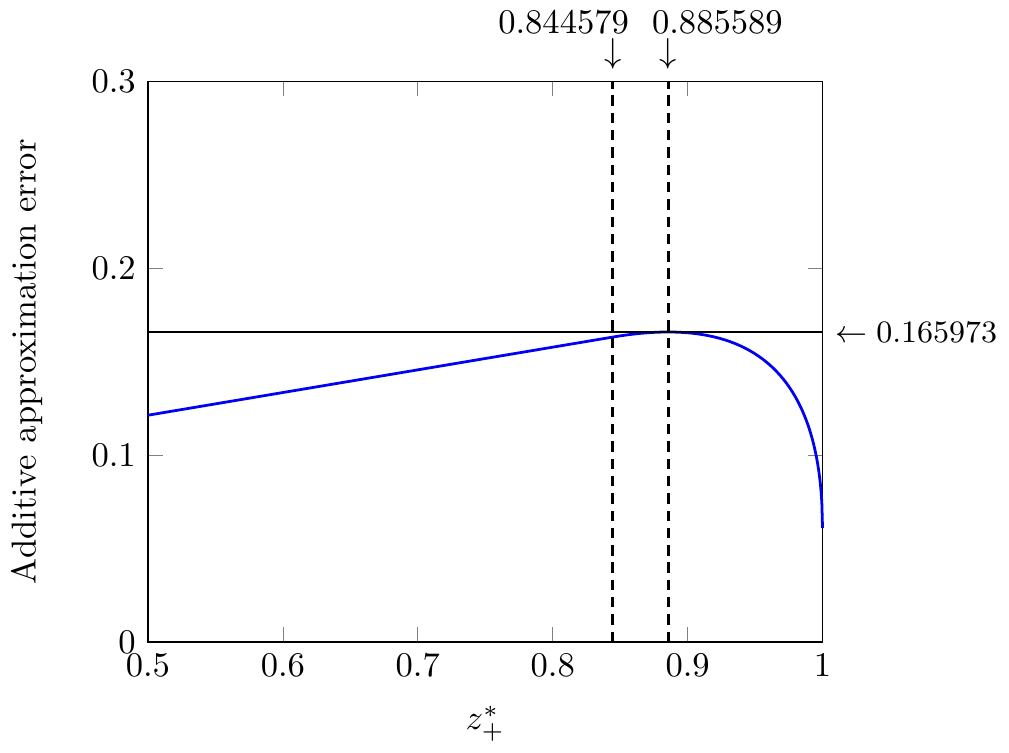}
\caption{An illustration of the additive approximation error of Algorithm~\ref{alg:cut} with respect to the value of $z^*_+$. 
Specifically, the function $g(x)=x-\underline{p_+}(2x-1)-(\alpha-1)/2$ for $x\in [1/2,1]$ is plotted. 
The maximum value of $g(x)$ is attained at $x= \frac{1}{2}+ \frac{\sqrt{\pi^2-4}}{2\pi}\ (\simeq 0.885589)$.
}\label{fig:error_cut}
\end{figure}
By simple calculation, we have 
\begin{align*}
\max_{1/2\leq x\leq 1} g(x) = \frac{\sqrt{\pi^2-4}}{2\pi} + \frac{\arccos\left(\frac{\sqrt{\pi^2-4}}{\pi}\right)}{\pi}-\frac{\alpha}{2},
\end{align*}
which is attained at $x= \frac{1}{2}+ \frac{\sqrt{\pi^2-4}}{2\pi}\ (\simeq 0.885589)$.
\end{proof}


Finally, we present a lower bound on $\text{E}[Q(\mathcal{C}_\text{out})]$ with respect to the value of $\OPTc$ (rather than $z^*_+$). 
Specifically, we have the following theorem. 
\begin{theorem}\label{thm:mod_cut}
It holds that 
\begin{align*}
\text{E}[Q(\mathcal{C}_\text{out})]> \OPTc -0.16598.
\end{align*}
In particular, if $\OPTc \geq \frac{\sqrt{\pi^2-4}}{2\pi}\ (\simeq 0.385589)$ holds, then
\begin{align*}
\text{E}[Q(\mathcal{C}_\text{out})]\geq \OPTc -g(\OPTc + 1/2).
\end{align*}
\end{theorem}

\begin{proof}
From Lemmas~\ref{lem:lb_opt_cut} and \ref{lem:max_cut}, it follows immediately that 
\begin{align*}
\text{E}[Q(\mathcal{C}_\text{out})]> \OPTc -0.16598.
\end{align*}

Here we prove the remaining part of the theorem. 
Assume that $\OPTc \geq \frac{\sqrt{\pi^2-4}}{2\pi}$ holds.
Clearly, the function $g(x)$ is monotonically decreasing over the interval $\left[\frac{1}{2}+\frac{\sqrt{\pi^2-4}}{2\pi},\, 1\right]$. 
Therefore, we have 
\begin{align*}
\text{E}[Q(\mathcal{C}_\text{out})]&\geq \OPTc -g(z^*_+)\geq \OPTc -g(\OPTc + 1/2), 
\end{align*}
where the second inequality follows from $z^*_+\geq \OPTc +1/2$. 
\end{proof}

Figure~\ref{fig:lb_cut} depicts the above lower bound on $\text{E}[Q(\mathcal{C}_\text{out})]$. 
As can be seen, if $\OPTc$ is close to $1/2$, then Algorithm~\ref{alg:cut} obtains a nearly-optimal solution. 

\begin{figure}[tb]
\centering
\includegraphics[width=9.3cm]{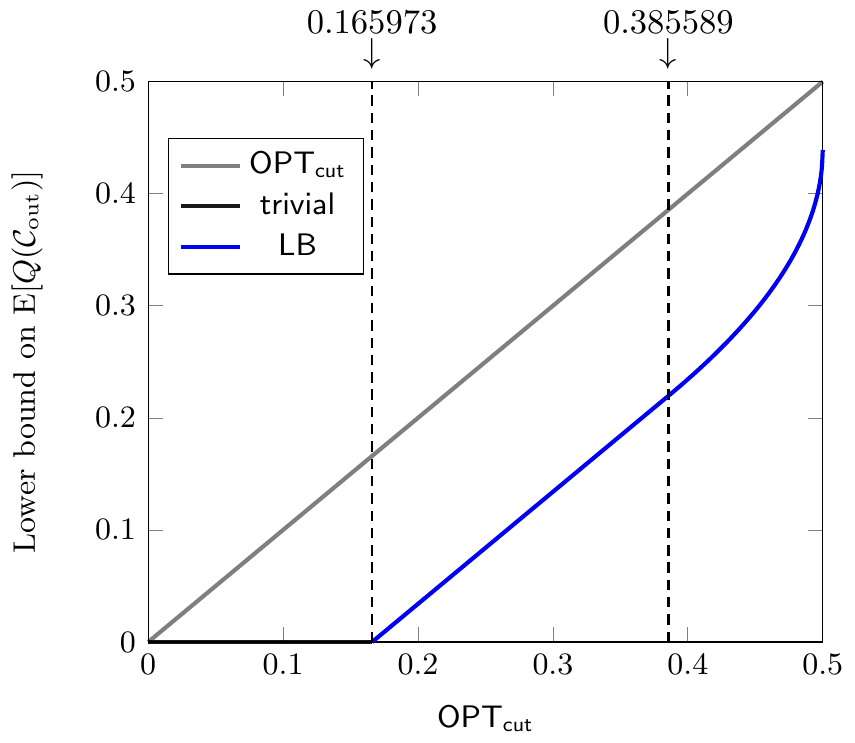}
\caption{An illustration of the lower bound on the expected modularity value of the output of Algorithm~\ref{alg:cut}.}\label{fig:lb_cut}
\end{figure}

\section{Related Problems}\label{sec:related}
In this section, we demonstrate that our additive approximation algorithm for the modularity maximization problem 
can be extended to some related problems.

\paragraph{Modularity maximization on edge-weighted graphs.}
First, we consider community detection in edge-weighted graphs. 
Let $G=(V,E,w)$ be a weighted undirected graph consisting of $n=|V|$ vertices, $m=|E|$ edges, and weight function $w:E\rightarrow \mathbb{R}_{>0}$. 
For simplicity, let $w_{ij}=w(i,j)$ and $W=\sum_{\{i,j\} \in E} w_{ij}$.
The weighted modularity, which was introduced by Newman~\cite{Ne04}, can be written as 
\begin{align*}
Q_w(\mathcal{C}) = \frac{1}{2W}\sum_{i\in V}\sum_{j\in V}\left(w_{ij} - \frac{s_is_j}{2W}\right)\delta(\mathcal{C}(i), \mathcal{C}(j)),
\end{align*}
where $s_i$ represents the weighted degree of $i\in V$ (i.e., $s_i=\sum_{\{i,j\}\in E} w_{ij}$). 
We consider the weighted modularity maximization problem: given a weighted undirected graph $G=(V,E,w)$, 
we are asked to find a partition $\mathcal{C}$ of $V$ that maximizes the weighted modularity. 
Since this problem is a generalization of the modularity maximization problem, it is also NP-hard.  

Our additive approximation algorithm for the modularity maximization problem can be generalized to the weighted modularity maximization problem. 
In fact, it suffices to set 
\begin{align*}
q_{ij}=\frac{w_{ij}}{2W}-\frac{s_is_j}{4W^2}\quad \text{for each } i,j\in V.  
\end{align*}
The analysis of the additive approximation error is similar; thus we have the following corollary. 
\begin{corollary}
There exists a polynomial-time $\left(\cos\left(\frac{3-\sqrt{5}}{4}\pi\right)-\frac{1+\sqrt{5}}{8}\right)$-additive approximation algorithm  for the weighted modularity maximization problem. 
\end{corollary}

\paragraph{Modularity maximization on directed graphs.}
Next we consider community detection in directed graphs. 
Let $G=(V,A)$ be a directed graph consisting of $n=|V|$ vertices and $m=|A|$ edges. 
The directed modularity, which was introduced by Leicht and Newman~\cite{LeNe08}, can be written as 
\begin{align*}
Q_d(\mathcal{C})=\frac{1}{m}\sum_{i\in V}\sum_{j\in V}\left(A_{ij}-\frac{d^\text{out}_id^\text{in}_j}{m}\right)\delta(\mathcal{C}(i),\mathcal{C}(j)), 
\end{align*}
where $A_{ij}$ is the $(i,j)$ component of the (directed) adjacency matrix $A$ of $G$, 
and $d^\text{in}_i$ and $d^\text{out}_i$, respectively, represent the in- and out-degree of $i\in V$. 
Note that there is no factor of $2$ in the denominators, unlike the (undirected) modularity. 
This is due to the directed counterpart of the null model used in the definition~\cite{LeNe08}. 
We consider the directed modularity maximization problem: given a directed graph $G=(V,A)$, 
we are asked to find a partition $\mathcal{C}$ of $V$ that maximizes the directed modularity. 
As mentioned in DasGupta and Desai~\cite{DaDe13}, this problem is also a generalization of the modularity maximization problem, and thus NP-hard. 

Our additive approximation algorithm for the modularity maximization problem can also be generalized to the directed modularity maximization problem. In fact, it suffices to set 
\begin{align*}
q_{ij}=\frac{A_{ij}}{m}-\frac{d^\text{out}_id^\text{in}_j}{m^2}\quad \text{for each } i,j\in V. 
\end{align*}
The analysis of the additive approximation error is also similar; thus we have the following corollary. 
\begin{corollary}
There exists a polynomial-time $\left(\cos\left(\frac{3-\sqrt{5}}{4}\pi\right)-\frac{1+\sqrt{5}}{8}\right)$-additive approximation algorithm  for the directed modularity maximization problem. 
\end{corollary}

\paragraph{Barber's bipartite modularity maximization.}
Finally, we consider community detection in bipartite graphs. 
Let $G=(V,E)$ be an undirected bipartite graph consisting of $n=|V|$ vertices and $m=|E|$ edges, 
where $V$ can be divided into $V_1$ and $V_2$ so that every edge in $E$ has one endpoint in $V_1$ and the other in $V_2$. 
Although the modularity is applicable to community detection in bipartite graphs, 
the null model used in the definition does not reflect the structural property of bipartite graphs. 
Thus, if we know that the input graphs are bipartite, 
the modularity is not an appropriate quality function. 

To overcome this concern, Barber~\cite{Ba07} introduced a variant of the modularity, which is called the bipartite modularity, 
for community detection in bipartite graphs. 
The bipartite modularity can be written as 
\begin{align*}
Q_b(\mathcal{C})=\frac{1}{m}\sum_{i\in V_1}\sum_{j\in V_2}\left(A_{ij}-\frac{d_id_j}{m}\right)\delta(\mathcal{C}(i),\mathcal{C}(j)). 
\end{align*}
Note again that there is no factor of $2$ in the denominators. 
This is due to the bipartite counterpart of the null model used in the definition~\cite{Ba07}. 
We consider Barber's bipartite modularity maximization problem: given an undirected bipartite graph $G=(V,E)$, 
we are asked to find a partition $\mathcal{C}$ of $V$ that maximizes the bipartite modularity. 
This problem is known to be NP-hard~\cite{MiSu15}.

Our additive approximation algorithm for the modularity maximization problem is applicable to Barber's bipartite modularity maximization problem. For each $i,j\in V$, we set 
\begin{align*}
q_{ij}=
\begin{cases}
\displaystyle \frac{A_{ij}}{m}-\frac{d_id_j}{m^2}  & \text{if } i\in V_1 \text{ and } j\in V_2,\\
0  &\text{otherwise}.
\end{cases}
\end{align*}
The analysis of the additive approximation error is again similar; thus we have the following corollary. 
\begin{corollary}
There exists a polynomial-time $\left(\cos\left(\frac{3-\sqrt{5}}{4}\pi\right)-\frac{1+\sqrt{5}}{8}\right)$-additive approximation algorithm  for Barber's bipartite modularity maximization problem. 
\end{corollary}

\section{Conclusions}\label{sec:conclusions}
In this study, we have investigated the approximability of modularity maximization. 
Specifically, we have proposed a polynomial-time 
$\left(\cos\left(\frac{3-\sqrt{5}}{4}\pi\right)-\frac{1+\sqrt{5}}{8}\right)$-additive approximation algorithm for the modularity maximization problem. 
Note here that $\cos\left(\frac{3-\sqrt{5}}{4}\pi\right)-\frac{1+\sqrt{5}}{8}<0.42084$ holds; 
thus, this improves the current best additive approximation error of 0.4672, which was recently provided by Dinh, Li, and Thai~\cite{DiLiTh15}. 
Interestingly, our analysis has also demonstrated that 
the proposed algorithm obtains a nearly-optimal solution for any instance with a very high modularity value. 
Moreover, we have proposed a polynomial-time 0.16598-additive approximation algorithm for the maximum modularity cut problem. 
It should be noted that this is the first non-trivial approximability result for the problem. 
Finally, we have demonstrated that our additive approximation algorithm for the modularity maximization problem can be extended to some related problems.

There are several directions for future research. 
It is quite interesting to investigate additive approximation algorithms for the modularity maximization problem more deeply. 
For example, it is challenging to design an algorithm that has a better additive approximation error than that of \textsf{Hyperplane}($k^*$). 
As another approach, is it possible to improve the additive approximation error of \textsf{Hyperplane}($k^*$) by completely different analysis? 
Our analysis implies that if we lower bound the expectation $\text{E}[Q(\mathcal{C}_k)]$ by the form in Lemma~\ref{lem:lb_opt}, 
our additive approximation error of $\cos\left(\frac{3-\sqrt{5}}{4}\pi\right)-\frac{1+\sqrt{5}}{8}$ is the best possible. 
As another future direction, the inapproximability of the modularity maximization problem in terms of additive approximation should be investigated, 
as mentioned in Dinh, Li, and Thai~\cite{DiLiTh15}. 
Does there exist some constant $\epsilon>0$ such that computing an $\epsilon$-additive approximate solution for the modularity maximization problem is NP-hard?

\section*{Acknowledgments}
YK is supported by a Grant-in-Aid for Research Activity Start-up (No.~26887014).  
AM is supported by a Grant-in-Aid for JSPS Fellows (No.~26-11908).

\bibliographystyle{abbrv}
\bibliography{mod}

\appendix
\section{Proof of Lemma~\ref{lem:compromise}}\label{apx:compromise}
We start with the following fact. 
\begin{fact}\label{fact:f1_bound}
If $\OPT <\cos\left(\frac{3-\sqrt{5}}{4}\pi\right)$ holds, then $f_1(\OPT)<\frac{1+\sqrt{5}}{4}$. 
Conversely, if $\OPT \geq \cos\left(\frac{3-\sqrt{5}}{4}\pi\right)$ holds, then $f_1(\OPT)\geq \frac{1+\sqrt{5}}{4}$. 
\end{fact}
\begin{proof}
Assume that $\OPT <\cos\left(\frac{3-\sqrt{5}}{4}\pi\right)$ holds. Then, we have  
\begin{align*}
f_1(\OPT)= 1-\frac{\arccos(\OPT)}{\pi}< 1-\frac{\arccos\left(\cos\left(\frac{3-\sqrt{5}}{4}\pi\right)\right)}{\pi} =\frac{1+\sqrt{5}}{4}. 
\end{align*}
The proof of the second statement is similar. 
\end{proof}

First, we consider the case where \(\OPT<\cos\left(\frac{3-\sqrt{5}}{4}\pi\right)\) holds.
Let us fix an arbitrary integer $k\geq 3$. 
We define the function \(l_k(x)=\frac{1}{2^{k-1}}\sum_{i=0}^{k-1}(2x)^i\) for $x\in [0,1]$. 
Then, we have
\begin{align*}
l_k(0)=0<0+\frac{1}{2}\ \  \text{and}\ \ 
l_k\left(\frac{1+\sqrt{5}}{4}\right)=\frac{1}{2^{k-1}}\frac{\left(\frac{1+\sqrt{5}}{2}\right)^k-1}{\frac{1+\sqrt{5}}{2}-1}
\leq \frac{1}{2^{3-1}}\frac{\left(\frac{1+\sqrt{5}}{2}\right)^3-1}{\frac{1+\sqrt{5}}{2}-1} =\frac{1+\sqrt{5}}{4}+\frac{1}{2}. 
\end{align*}
Moreover, the function $l_k(x)$ is convex. 
Thus, we have $l_k(x)<x+1/2$ for any $x\in \left[0, \frac{1+\sqrt{5}}{4}\right)$. 
Here we evaluate the value of 
\begin{align*}
g_k(\OPT)-g_2(\OPT) 
&=(\OPT -f_k(\OPT)+1/2^k) - (\OPT -f_2(\OPT)+1/4)\\
&=(f_2(\OPT)-1/4) - (f_k(\OPT)-1/2^k)\\
&=\left(f_1(\OPT)-1/2\right)\left(f_1(\OPT)+1/2\right) - \left(f_1(\OPT)-1/2\right)\cdot l_k(f_1(\OPT))\\
&=\left(f_1(\OPT)-1/2\right)\left(\left(f_1(\OPT)+1/2\right)-l_k(f_1(\OPT))\right).
\end{align*}
From the definition, \(f_1(\OPT)\geq 1/2\) holds. By Fact~\ref{fact:f1_bound}, we get $f_1(\OPT)+1/2> l_k(f_1(\OPT))$. 
Therefore, we have \(g_k(\OPT)\geq g_2(\OPT)\), 
which means that $\argmin_{k\in \mathbb{Z}_{>0}}g_k(\OPT)$ contains $1$ or $2$. 

Next, we consider the case where \(\OPT\ge \cos\left(\frac{3-\sqrt{5}}{4}\pi\right)\) holds.
It suffices to show that for any integer $k\geq \log_{2}n$,
\begin{align*}
g_{k+1}(\OPT)>g_k(\OPT). 
\end{align*}
In addition to Fact~\ref{fact:f1_bound}, we have the following two facts. 
\begin{fact}[see Lemma~2.1 in \cite{DaDe13}]\label{fact:OPT_ub}
It holds that $\OPT \leq 1-1/n$. 
\end{fact}

\begin{fact}\label{fact:acos_lb}
For any $x\in [0,1]$, it holds that $\sqrt{2x}\leq \arccos{(1-x)}$. 
\end{fact}
\begin{proof}
By Jordan's inequality (i.e., $\sin t\leq t$ for any $t\in [0,\pi/2]$), we have 
\begin{alignat*}{3}
&                &\quad &\int_0^y\sin t\,dt\le \int_0^y t\,dt &\quad &(\forall y\in[0,\pi/2])\\
&\Leftrightarrow &      &1-\cos y\le \frac{y^2}{2}        &      &(\forall y\in[0,\pi/2])\\
&\Leftrightarrow &      &1-\frac{y^2}{2}\le \cos y        &      &(\forall y\in[0,\pi/2])\\
&\Leftrightarrow &      &1-\frac{\arccos(z)^2}{2}\le z    &      &(\forall z\in[0,1])\\
&\Leftrightarrow &      &2(1-z)\le \arccos(z)^2           &      &(\forall z\in[0,1])\\
&\Leftrightarrow &      &\sqrt{2x}\le \arccos(1-x)        &      &(\forall x\in[0,1]). 
\end{alignat*}
\end{proof}
Using the above facts, we prove the inequality.  
For any integer $k\geq \log_{2}n$, we have 
\begin{alignat*}{3}
&g_{k+1}(\OPT)-g_k(\OPT) &&=(\OPT-f_{k+1}(\OPT)+1/2^{k+1}) -(\OPT-f_k(\OPT)+&&1/2^k)\\
&&&= (1-f_1(\OPT))f_k(\OPT) -1/2^{k+1}\\
&&&= \frac{\arccos(\OPT)}{\pi} f_k(\OPT) -1/2^{k+1}\\
&&&\geq \frac{\arccos(1-1/n)}{\pi} f_k(\OPT) -1/2^{k+1}                            &&(\text{By Fact~\ref{fact:OPT_ub}})\\
&&&\geq \frac{\sqrt{2}}{\pi\sqrt{n}} f_k(\OPT)-1/2^{k+1}                           &&(\text{By Fact~\ref{fact:acos_lb}})\\
&&&=\frac{1}{2^{k+1}}\left(\frac{2\sqrt{2}}{\pi\sqrt{n}} (2\cdot f_1(\OPT))^k-1\right)\\
&&&\geq \frac{1}{2^{k+1}}\left(\frac{2\sqrt{2}}{\pi\sqrt{n}} \left(\frac{1+\sqrt{5}}{2}\right)^k-1\right)  &&(\text{By Fact~\ref{fact:f1_bound}})\\
&&&> \frac{1}{2^{k+1}}\left(\frac{2\sqrt{2}}{\pi\sqrt{n}} \left(2^{2/3}\right)^k-1\right)\\
&&&\geq \frac{1}{2^{k+1}}\left(\frac{2\sqrt{2}}{\pi}n^{1/6}-1\right)\\
&&&>0, 
\end{alignat*}
where the last inequality follows from $n\geq 2$ by the assumption that \(\OPT\ge \cos\left(\frac{3-\sqrt{5}}{4}\pi\right)\) holds. 
\qed

\end{document}